\documentclass[12pt]{article}

\usepackage{amsmath,amsfonts,amssymb,amsthm,harvard}

\def\R{\mathbb{R}}

\def\N{\mathbb{N}}

\def\e{\varepsilon}

\def\1{\mathds{1}}

\def\M{\mathcal{M}}
\def\P{\mathcal{P}}
\def\Q{\mathcal{Q}}

\def\MA{M(X,\Sigma)}
\def\M1{M_1(X,\Sigma)}
\def\V{W}
\def\m{\mathfrak{m}}

\def\Hra{\mathop{H_\alpha^R}}
\def\hra{\mathop{h_\alpha^R}}
\def\Hta{\mathop{H_\alpha^T}}
\def\hta{\mathop{h_\alpha^T}}
\def\Hs{\mathop{H^S}}
\def\hs{\mathop{h^S}}
\def\H{\mathop{H}}
\def\h{\mathop{h}}

\def\htva{\mathop{hv_{\alpha}^T}}

\def\hsv{\mathop{hv^S}}
\def\Hsv{\mathop{Hv^S}}
\def\hv{\mathop{hv}}
\def\Hv{\mathop{Hv}}

\def\for{\mbox{  for }}

\newtheorem{observation}{Observation}[section]
\newtheorem{theo}{Theorem}[section]
\newtheorem{prop}{Proposition}[section]

\theoremstyle{definition}
\newtheorem{definition}{Definition}[section]

\newtheorem{example}{Example}[section]

\title{Weighted Approach to General Entropy Function}
\author{Marek \'{S}mieja}
\date{Institute of Computer Science\\
Department of Mathematics and Computer Science\\
Jagiellonian University\\
Lojasiewicza 6, 30-348, Krakow, Poland\\
Email: marek.smieja@ii.uj.edu.pl}
\begin{document}

\maketitle

\begin{abstract}
The definition of weighted entropy allows for easy calculation of the entropy of the mixture of measures. In this paper we investigate the problem of equivalent definition of the general entropy function in weighted form. We show that under reasonable condition, which is satisfied by the well-known Shannon, R\'enyi and Tsallis entropies, every entropy function can be defined equivalently in the weighted way. As a corollary, we show how use the weighted form to compute Tsallis entropy of the mixture of measures.
\end{abstract}

\textbf{Keywords: } entropy, weighted entropy, mixture of measures, coding, data compression 

\newpage

\section{Introduction}
The entropy is an important tool used to examine and analyze the behavior of statistical and physical systems. It is widely applied in information theory, thermodynamics, quantum mechanics and many others fields of science (see \citeasnoun{El}, \citeasnoun{Go}, \citeasnoun{Se}, \citeasnoun{Wu}, \citeasnoun{Fra} and \citeasnoun{Li}).

There are many kinds of the entropy functions. One of the most popular is the Shannon entropy \cite{Sh}. Given a probability measure $\mu$ on data space $X$ a countable partition $\P$, the Shannon entropy of $\P$ is defined by:
$$
\hs(\mu;\P)=-\sum_{P \in \P} \mu(P) \log_2 (\mu(P)).
$$
It determines the statistical amount of memory used in lossy coding elements of $X$ by the elements of partition $\P$. In our consideration we use more general notion of the Shannon entropy which is based on R\'enyi's idea of entropy dimension \cite{Re2}. Given a measurable cover $\Q$ of $X$, we define $\Hs(\mu;\Q)$ as an infimum of the entropies taken over all partitions finer than $\Q$ (see the next section).

We have recently propose an equivalent weighted approach to the Shannon entropy which is based on measures instead of partitions \cite{Sm}. It can be seen as a horizontal splitting of data space in contrast to a classical vertical one. Roughly speaking, given a division of measure $\mu$ into ``submeasures'' $(\mu_i)_{i \in \N}$ (i.e. $\mu=\sum\limits_{i \in \N} \mu_i(X)$), we rewritten the entropy in terms of measures as:
$$
\hsv(\mu;(\mu_i)_{i \in \N})=-\sum_{i \in \N} \mu_i(X) \log_2 (\mu_i(X)).
$$
This reformulation allows to replace the undefined operation on partitions $\P_1+\P_2$ into well defined operation on functions $\mu_1 + \mu_2$. It is extremely useful when computing the entropy of the mixture of measures. From practical point of view this approach describes the idea of random lossy coding.

Others popular entropy functions are the R\'enyi and Tsallis entropy of order $\alpha$ (see \citeasnoun{Re1} and \citeasnoun{Tsa}). They were created as the one parametric family of generalized entropy functions (precise definitions are given in the next section). The dependence of parameter $\alpha$, allows to weaken or emphasize some probability events. If $\alpha > 1$ then the entropy is more sensitive on events that occur often while for $\alpha \in (0,1)$ the entropy is more sensitive on the events that happen seldom \cite{Mas}. These both entropies generalize the Shannon entropy.

In this paper, we apply the idea of weighted entropy for various kinds of entropy functions. For this purpose we define a condition under which every entropy function can be equivalently defined in the weighted way. More precisely, let $\P$ be a partition and let $f:\R \to \R$ and $g:[0,1] \to \R$ be continuous functions. We say that the function of the form:
$$
\h(\mu;\P):=f(\sum_{P \in \P} g(\mu(P)) ),
$$
satisfies the condition of general entropy function if $f$ is increasing, $g$ is subadditive and concave or $f$ is decreasing, $g$ is superadditive and convex. It is easy to see that Shannon, R\'enyi and Tsallis entropies satisfy the above condition. Hence the weighted approach is equivalently defined for these kinds of entropies. 

As it was mentioned, the weighted approach allows to rewritten the entropy in the form of the function of measures instead of partition. It is useful in calculation the entropy of the combination of measures. As an example of the application of the weighted entropy, we show in this paper how use the weighted form for obtaining the estimation of Tsallis entropy of the mixture of measures (see Theorem \ref{corEnt}). We prove that calculated bounds are sharp. 

\section{Weighted Approach to Entropy.}

In this section we will show how apply the weighted approach to the general entropy function. Before that, let us recall the definition of weighted Shannon entropy \cite{Sm} to get an idea of weighted entropy. From now on, if not stated otherwise, we always assume that $(X, \Sigma, \mu)$ is a probability space.

The entropy is defined on the partition of data space $X$. We say that a family $\P \subset \Sigma$ is a partition of $X$ if $\P$ is countable family of disjoint sets and
$$
\mu(X \setminus \bigcup_{P \in \P}P) = 0.
$$
The Shannon entropy is defined as follows:
\begin{definition}
Let $\P$ be a partition of $X$. The Shannon entropy of $\P$ is given by:
$$
\hs(\mu;\P):= - \sum_{P \in \P} \mu(P) \log_2(\mu(P)).
$$
\end{definition}

If we consider the problem of lossy data compression then the partition is interpreted as a coding alphabet. We map every point $x \in X$ to unique $P \in \P$ such that $x \in P$. The entropy determines a statistical amount of memory per one element used in the lossy coding generated by partition $\P$.

Based on the R\'enyi idea of entropy dimension we generalize the Shannon entropy on the case of any measurable cover of data space. We say that one family $\P$ of subsets of $X$ is finer that the second family $\Q$ iff for every $P \in \P$ there exists $Q \in \Q$ such that $P \subset Q$. 
\begin{definition} \label{entrOgol}
Given a measurable cover $\Q$ of $X$ the Shannon entropy of $\Q$ is
$$
\Hs(\mu;\Q):=\inf\{\hs(\mu;\P) \in [0,\infty] : \mbox{$\P$ is a partition and $\P \prec \Q$}\}.
$$
\end{definition}
Clearly, if there is no partition finer than $\Q$, then $\H(\mu;\Q)=\infty$, as $\inf(\emptyset)=\infty$.

In the case of coding, cover $\Q$ defines the maximal error of the compression. We allow only codings with use of partitions which are finer than $\Q$ (we say then that partition is $\Q$-acceptable). The entropy describes the best lossy coding determined by $\Q$-acceptable alphabets.

One of the simplest error control family in metric space consists of all balls with given radius or cubes with specific edge length. Such notions were used by A. R\'enyi \cite{Re2} or E. C. Posner \cite{Po2}. Our approach allows to differ the size of particular sets from the partition. Intuitively, more probable events should be coded with smaller sets while less probable with bigger.

The inspiration of weighted entropy, lies in the horizontal partitioning of data space instead of classical vertical one. We substitute the division of space $X$ into partition by the division of measure $\mu$ into ``submeasures''\footnote{The idea of weighted entropy is indebted to the notion of weighted Hausdorff measures considered by J. Howroyd \cite{Ho1}.}. Roughly speaking, this approach provides the computation and interpretation of the entropy with respect to ``formal'' convex combination $a_1\P_1+a_2\P_2$, where $\P_1,\P_2$ are partitions.

We denote the division of measure $\mu$ with respect to $\Q \subset \Sigma$ by:
\begin{equation} \label{ZbiorV}
\begin{array}{ll}
\V(\mu;\Q):= & \{\m:\Q \ni Q \rightarrow \m_Q \in \MA : \\[0.4ex]
& \m_Q(X \setminus Q) = 0 \text{ for every $Q \in \Q$ and } \sum_{Q \in \Q} \m_Q = \mu \},
\end{array}
\end{equation}
where $\MA$ is the family of all measures on $(X, \Sigma)$. Observe that every function $\m \in \V(\mu;\Q)$ is non-zero on at most countable number of sets of $\Q$. We define the weighted Shannon entropy:
\begin{definition} 
The weighted Shannon entropy of a given $\m \in \V(\mu;\Q)$ by:
\begin{equation} \label{ShWei}
\hsv(\mu; \m):= - \sum_{Q \in \Q} \m_Q(X) \log_2(\m_Q(X)) \text{,}
\end{equation}
while the weighted Shannon entropy of measurable cover $\Q$ of $X$ is
$$
\Hsv(\mu;\Q):=\inf\{\hsv(\mu;\m) \in [0,\infty] : \m \in \V(\mu;\Q)\}.
$$
\end{definition}
The sum in the formula (\ref{ShWei}) is taken over $Q \in \Q$ such that $\m_Q(X) > 0$.

We show in \cite[Theorem II.1]{Sm} that the weighted $\mu$-entropy of $\Q$ is equal to the classical one which in consequence allows to compute the entropy of the mixture of sources:
\medskip
\begin{flushleft} \textbf{Shannon entropy of the mixture\cite[Theorem III.1]{Sm}:} \emph{Let $a_1, a_2 \in [0,1]$ be such that $a_1 + a_2 = 1$. If $\mu_1, \mu_2$ are probability measures and $\Q \subset \Sigma$ then:}
$$
\Hs(a_1\mu_1+a_2\mu_2;\Q) \geq a_1 \Hs(\mu_1;\Q)+a_2\Hs(\mu_2;\Q)  
$$
\emph{and}
$$
\Hs(a_1\mu_1+a_2\mu_2;\Q) \leq a_1 \Hs(\mu_1;\Q)+a_2\Hs(\mu_2;\Q) - a_1 \log_2(a_1) - a_2 \log_2(a_2).
$$
\medskip
\end{flushleft}

We distinguish other kinds of the entropy - R\'enyi and Tsallis entropies. For the convenience of the reader, we give their definitions.
\begin{definition}
Let $\alpha \in (0,\infty) \setminus \{1\}$ and let $\P$ be a partition of $X$. The R\'enyi entropy of $\P$ is
$$
{\hra}(\mu;\P):=\frac{1}{1-\alpha} \log_2[\sum_{P \in \P} \mu(P)^\alpha]
$$
and the Tsallis entropy of $\P$ is
$$
{\hta}(\mu;\P):=\frac{1}{1-\alpha} (\sum_{P \in \P} \mu(P)^\alpha - 1).
$$
These definitions are naturally generalized for any measurable cover $\Q \subset \Sigma$ as in Definition \ref{entrOgol}. We denotes these quantities by $\Hra(\mu;\Q)$ and $\Hta(\mu;\Q)$, respectively.
\end{definition}
For review of other kinds of information measures see books by J. N. Kapur \cite{Kap} and C. Arndt \cite{Arn}.

We move to the definition of weighted entropy for general entropy function. Let us first define what we are mean by the general entropy function.
\begin{definition}
Let $\P$ be a partition of $X$ and let $f:\R \to \R$ and $g:[0,1] \to \R$ be continuous functions. We say that the function of the form:
$$
\h(\mu;\P):=f(\sum_{P \in \P} g(\mu(P)) ),
$$
satisfies the condition of general entropy function (CGEF) if one of the following conditions is valid:
\begin{enumerate}
\item $f$ is increasing, $g$ is subadditive and concave,
\item $f$ is decreasing, $g$ is superadditive and convex.
\end{enumerate}
\end{definition}

The classical examples of the entropy function which satisfy the above condition are Shannon, R\'enyi and Tsallis entropies:
\begin{observation} 
Shannon, R\'enyi and Tsallis entropies satisfy the condition of general entropy function:
\begin{itemize}
\item For Shannon entropy, we have $f(x)=-x$ and $g(x)=x \log_2(x)$; 
\item For R\'enyi entropy, we have $g(x)=\frac{1}{1-\alpha} \log_2(x)$ and $g(x)=x^\alpha$ where $\alpha \in (0,\infty)\setminus\{1\}$;
\item For Tsallis entropy, we have $g(x)=\frac{1}{1-\alpha} (x - 1)$ and $g(x)=x^\alpha$ where $\alpha \in (0,\infty)\setminus\{1\}$.
\end{itemize}
\end{observation}
The CGEF will be crucial to define the form of weighted entropy equivalent to the classical one. Let us assume that the entropy function $\h$ satisfies the CGEF. Then given the error-control family $\Q$, we can rewritten it as follows:
$$
\hv(\mu;\m)=f(\sum_{Q \in \Q} g(\m(\Q))) \text{, for } \m \in \V(\mu;\Q).
$$
This is the weighted form of the entropy function $\h$. The generalized version of the weighted entropy $\H$ of $\Q$ is
$$
\Hv(\mu;\Q)=\inf\{\hv(\mu;\m) \in [0,\infty] : \mbox{ for } \m \in \V(\mu;\Q)\}.
$$
The weighted entropy provides the form of the entropy as a function of measures instead of partition. It is useful when computing the entropy of the combination of measures..

In the next section we show that weighted definition of the entropy is equivalent to the classical one if the entropy function satisfies the CGEF. As an example of the application of weighted entropy, in Section \ref{TsaSec} we estimate the Tsallis entropy of the combination of measure.

%%%%%%%%%%%%%%%%%%%%%%%%%%%%%%%%%%%%%%%%%%%%%%%%

\section{Equivalence between classical and weighted entropy}

We show the equivalence between classical and weighted form of the entropy under the CGEF. The proofs are based on the idea introduced in \cite{Sm}. To derive the equality we will show two inequalities.
\begin{prop} \label{prop}
Let $\Q \subset \Sigma$ and let $\h$ be the entropy function that satisfies the CGEF. Then
$$
{\Hv}(\mu; \Q) \leq {\H}(\mu; \Q).
$$
\end{prop}
\begin{proof}
Let us first observe that if there is no $\mu$-partition finer than $\Q$ then ${\H}(\mu;\Q)=\infty$ and the inequality holds.

Thus, we assume that $\P$ is a $\mu$-partition finer than $\Q$. We construct a function $\m \in \V(\mu;\Q)$ with not grater entropy than $\P$. 

Let us notice that, there exists a mapping $\pi: \P \to \Q$ such that $P \subset \pi(P)$ since $\P \prec \Q$. Next, we put 
$$
\P_\Q := \{P_Q\}_{Q \in \Q},
$$
where $P_Q:=\bigcup\limits_{P:\pi(P)=Q}P$. Finally, we obtain a function $\m:\Q \ni Q \rightarrow \mu_{|P_Q} \in \MA$.

Our aim is to verify that $\m \in \V(\mu; \Q)$. It is easy to see that $\P_\Q$ is a $\mu$-partition and $P_Q \subset Q$, for every $Q \in \Q$. Thus, we have
$$
\sum_{Q \in \Q} \m_Q(X) = \sum_{Q \in \Q} \mu_{|P_Q}(Q) 
= \sum_{Q \in \Q} \mu(P_Q) = \mu(X).
$$
The above sums are taken only over $Q \in \Q$ such that $\m_Q(X) > 0$. Moreover, we get
$$
\m_Q(X \setminus Q) = \mu_{|P_Q}(X \setminus Q) 
\leq \mu_{|Q}(X \setminus Q)= 0,
$$
for $Q \in \Q$. We conclude that $\m \in \V(\mu;\Q)$. 

We would like to check that the entropy of $\m$ is not grater than the entropy of $\P$. For this purpose, we use the CGEF:
$$
{\hv}(\mu; \m) 
= f \big(\sum_{Q \in \Q}g(\m_Q(X))\big) 
= f \big(\sum_{Q \in \Q} g( \mu_{|P_Q}(X))\big) 
$$
$$
= f \big(\sum_{Q \in \Q} g(\mu(P_Q))\big) 
= f\big(\sum_{Q \in \Q} g(\mu(\bigcup_{P: \pi(P)=Q}P))\big) 
$$
$$
\leq f\big(\sum_{Q \in \Q} \sum_{P: \pi(P)=Q} g(\mu(P))\big) 
= f\big(\sum_{P \in \P} g(\mu(P))\big)
={\h}(\mu; \P).
$$
Hence, we get that ${\Hv}(\mu; \Q) \leq {\H}(\mu; \Q)$.
\end{proof}

To derive the inequality ${\Hv}(\mu;\Q) \geq {\H}(\mu;\Q)$ we will apply Hardy Littlewood Polya Theorem. The version of Hardy Littlewood Polya Theorem for finite sequences is given in \cite[Theorem 1.5.4]{Ni} while the case of infinite sequences is presented in \cite[Appendix A]{Sm}. Let us recall this theorem:
\medskip
\begin{flushleft} Hardy Littlewood Polya Theorem. \emph{Let $a>0$ and let $\varphi:[0,a] \to (0,\infty)$, $\varphi(0)=0$ be a continuous function.
Let $(x_i)_{i \in I}, (y_i)_{i \in I} \subset [0,a]$ be given sequences where
either $I=\N$ or $I=\{1,\ldots,N\}$ for a certain $N \in \N$. We assume that $(x_i)_{i \in I}$ is a nonincreasing sequence and}
$$
\sum_{i=1}^n x_i \leq \sum_{i=1}^n y_i \for n \in I,
$$
$$
\sum_{i \in I}x_i=\sum_{i \in I}y_i.
$$
\emph{Then}
\begin{itemize}
\item $\sum_{i \in I} \varphi(x_i) \geq \sum_{i \in I} \varphi(y_j)$ if $\varphi$ is concave,
\item $\sum_{i \in I} \varphi(x_i) \leq \sum_{i \in I} \varphi(y_j)$ if $\varphi$ is convex.
\end{itemize}
\medskip
\end{flushleft}

Then the following proposition holds:
\begin{prop} \label{waz}
Let $\Q =\{Q_i\}_{i \in I} \subset \Sigma$, where either $I=\N$ or $I=\{1,\ldots,N\}$ for a certain $N \in \N$. Let $\m \in \V(\mu;\Q)$ and let $\h$ be the entropy function which satisfies the CGEF. We assume that
\begin{itemize}
 \item $\mu(X \setminus \bigcup\limits_{i \in I}Q_i) = 0$,
 \item the sequence $I \ni i \rightarrow \m_{Q_i}(X)$ is nonincreasing.
\end{itemize}
We define the family $\P=\{P_i\}_{i \in I} \subset \Sigma$ by the formula
$$
P_1:=Q_1, \, P_i:=Q_i \setminus \bigcup_{k=1}^{i-1} Q_{k} \for i \in I, i \geq 2.
$$
Then $\P$ is a $\mu$-partition, $\P \prec \Q$ and
\begin{equation} \label{cosik}
{\hv}(\mu; \m) \geq {\h}(\mu;\P).
\end{equation}
\end{prop}
\begin{proof}
Directly from the definition of family $\P$ and $\Q$, we get that $\P$ is $\Q$-acceptable partition. 

We prove the inequality (\ref{cosik}). To do this we will use Hardy Littlewood Polya Theorem. We define the sequences $(x_i)_{i \in I} \subset [0,1]$ and $(y_i)_{i \in I} \subset [0,1]$ by the formulas
$$
x_i:=\m_{Q_i}(X)=\m_{Q_i}(Q_i), \, 
y_i:=\mu(P_i)
$$
for $i \in I$.  

Clearly, $(x_i)_{i \in I}$ is nonincreasing and
$$
\sum_{i \in I} x_i=\mu(X)=\sum_{i \in I} y_i.
$$
Moreover, for every $n \in I$:
$$
\sum_{i = 1}^n x_i=\sum_{i = 1}^n \m_{Q_i}(Q_i) 
=(\sum_{i = 1}^n \m_{Q_i})(Q_1 \cup \ldots \cup Q_n)
$$
$$
\leq \mu(Q_1 \cup \ldots \cup Q_n)
=\sum_{i = 1}^n \mu(P_i)=\sum_{i = 1}^n y_i.
$$

Thus these sequences satisfy the assumptions of Hardy Littlewood Polya Theorem. Making use of CGEF, we conclude that
$$
{\hv}(\mu; \m) = f\big(\sum_{i \in I} g(\m_{Q_i}(X))\big) 
= f(\sum_{i \in I} g(x_i)) 
$$
$$
\geq f(\sum_{i \in I} g(y_i)) 
= f\big(\sum_{i \in I}g(\mu(P_i))\big)
={\h}(\mu;\P),
$$
which completes the proof.
\end{proof}

We are going to present the main result of this paper -- the equivalence between classical and weighted entropy under the CGEF.

\begin{theo} \label{wnWaz}
Let $\Q$ be an error-control family and let $\h$ be the entropy function which satisfies the CGEF. The weighted form of the entropy function $\H$ equals the classical one, i.e.
$$
{\Hv}(\mu;\Q) = {\H}(\mu;\Q).
$$
\end{theo}

\begin{proof}
We will show that ${\Hv}(\mu;\Q) \geq {\H}(\mu;\Q)$. The opposite inequality follows directly from Proposition \ref{prop}.

Let us first observe that if $\V(\mu;\Q) = \emptyset$ then ${\Hv}(\mu;\Q) = \infty$ and consequently the proof is completed since ${\Hv}(\mu;\Q) \geq {\H}(\mu;\Q)$. 

Thus let us assume that $\m \in \V(\mu;\Q)$. We construct the subset of family $\Q$ by:
$$
\tilde{\Q} := \{Q \in \Q : \m_Q(X) > 0\}.
$$
Clearly, $\tilde{\Q}$ is a countable family since $\sum\limits_{Q \in \tilde{\Q}} \m_Q(X) = 1$ and $\tilde{\m}:=\m_{|\tilde{\Q}} \in \V(\mu; \tilde{\Q})$. Moreover, $\tilde{\Q} \prec \Q$ and ${\hv}(\mu;\tilde{\m})={\hv}(\mu;\m)$.

As $\tilde{\Q}$ is countable, we may find a set of indices $I \subset \N$ such that $\tilde{\Q}=\{Q_i\}_{i \in I}$ and the sequence $I \ni i \rightarrow \m_{Q_i}(X)$ is nonincreasing. Making use of Proposition \ref{waz} we construct a $\mu$-partition $\P \prec \tilde{\Q}$, which satisfies
$$
{\hv}(\mu;\tilde{\m}) \geq {\h}(\mu;\P).
$$
This completes the proof since $\P \prec \tilde{\Q} \prec \Q$ and ${\hv}(\mu;\m) = {\hv}(\mu;\tilde{\m}) \geq {\h}(\mu;\P)$.
\end{proof}

%%%%%%%%%%%%%%%%%%%%%%%%%%%%%%%%%%%%%%%%%%%%%%%%%%%%%%%%%%%%%%%%%%%%%%%%%%%%%%%%

\section{Application of weighted form of the entropy} \label{TsaSec}

In this section we present how use the weighted form of the entropy in the calculation of the entropy of the mixture of measures. The reader interested in the topic of the mixture of measures is referred to \cite{Sm} where this problem is explain in details.

Our aim is to show how estimate the Tsallis entropy $\Hta$ of the mixture of measures in terms of the entropies of the individual measures. Let us start with the proposition:
\begin{prop} \label{propozycja}
We assume that $\alpha \in (0,\infty)\setminus\{1\}$ and $n \in \N$. Let $a_k \in (0,1)$ for $k \in \{1,\ldots,n\}$ be such that $\sum\limits_{k=1}^n a_k = 1$ and let $\{\mu_k\}_{k=1}^n$ be a family of probability measures. We define $\mu := \sum\limits_{k=1}^n a_k \mu_k$. 
\begin{itemize}	
\item
If $\P$ is a $\mu$-partition of $X$ then $\P$ is a $\mu_k$-partition of $X$ for $k \in \{1,\ldots,n\}$ and
\begin{equation} \label{pierwsza}
{\hta}(\mu; \P) \geq \sum_{i=1}^n a_i \hta(\mu_i;\P).
\end{equation}
\item If $\Q \subset \Sigma$ and $\m^k \in \V(\mu_k;\Q)$ for $k \in \{1,\ldots,n\}$ then $\m := \sum\limits_{k=1}^n a_k \m^k \in \V(\mu;\Q)$ and
\begin{equation} \label{druga}
{\htva}(\mu; \m) \leq \sum_{i=1}^n a_i^\alpha \htva(\mu_i;\m) + \frac{\sum_{k=1}^n a_k^\alpha - 1}{1-\alpha}.
\end{equation}
\end{itemize}
\end{prop}

\begin{proof}
Let us first observe that $\P$ is a $\mu_k$-partition of $X$, for every $k \in \{1,\ldots,n\}$. 

Then, making use of CGEF, we have
$$
{\hta}(\mu;\P)= \frac{1}{1-\alpha} \big[ \sum_{P \in \P}\big(\sum_{k=1}^n a_k\mu_k(P)\big)^\alpha - 1\big]
$$
$$
\geq \frac{1}{1-\alpha} \big[ \sum_{k=1}^n \big(a_k \sum_{P \in \P} \mu_k(P)^\alpha\big) - 1 \big]  
$$
$$
=\frac{1}{1-\alpha} \big[\sum_{k=1}^n a_k \big( \sum_{P \in \P} \mu_k(P)^\alpha - 1\big)\big]
$$
$$
=\sum_{k=1}^n a_k \hta(\mu_k;\P).
$$
It completes (\ref{pierwsza}).

We derive the second part of the proposition. Clearly, $\m \in \V(\mu;\Q)$. To see that (\ref{druga}) holds, we again use the CGEF:
$$
{\htva}(\mu;\m)= \frac{1}{1-\alpha} \big[ \sum_{Q \in \Q}\big(\sum_{k=1}^n a_k\m_k(Q)\big)^\alpha - 1\big]
$$
$$
\leq \frac{1}{1-\alpha} \big[ \sum_{k=1}^n \big(a_k^\alpha \sum_{Q \in \Q} \m_k(Q)^\alpha\big) - 1 \big]  
$$
$$
=\frac{1}{1-\alpha} \big[\sum_{k=1}^n a_k^\alpha \big( \sum_{Q \in \Q} \m_k(Q)^\alpha - 1 \big)\big] + \frac{\sum_{k=1}^n a_k^\alpha - 1}{1-\alpha}
$$
$$
=\sum_{k=1}^n a_k^\alpha \htva(\mu_k;\m) + \frac{\sum_{k=1}^n a_k^\alpha - 1}{1-\alpha}.
$$
\end{proof}

This result allows us to estimate the Tsallis entropy $\Hta$ of the mixture of measures.
\begin{theo} \label{corEnt}
Let $\alpha \in (0,\infty)\setminus\{1\}$ and $n \in \N$. We assume that $a_k \in [0,1]$ for $k \in \{1,\ldots,n\}$ be such that $\sum\limits_{k=1}^n a_k = 1$. Let $\{\mu_k\}_{k=1}^n$ be a family of probability measures and $\mu := \sum\limits_{k=1}^n a_k \mu_k$. If $\Q \subset \Sigma$ then
\begin{equation} \label{nierowWaz}
{\Hta}(\mu; \Q) \geq \sum_{i=1}^n a_i \Hta(\mu_i;\P).
\end{equation}
and
\begin{equation} \label{nierowWaz2}
{\Hta}(\mu; \Q) \leq \sum_{i=1}^n a_i^\alpha \Hta(\mu_i;\P) + \frac{\sum_{k=1}^n a_k^\alpha - 1}{1-\alpha}.
\end{equation}
\end{theo}

\begin{proof}
Let us first consider the case when ${\Hta}(\mu_k;\Q) = \infty$ for a certain $k\in\{1,\ldots,n\}$. Then also ${\Hta}(\mu;\Q) = \infty$ and the inequalities hold trivially. 

Thus let us assume that for every $k\in\{1,\ldots,n\}$, ${\Hta}(\mu_k;\Q) < \infty$. Without loss of generality, we may assume also that $a_k \neq 0$ for every $k \in \{1,\ldots,n\}$. Let $\e>0$ be arbitrary. 

To prove the first inequality, we find a $\mu$-partition $\P$ finer than $\Q$ such that 
\begin{equation} \label{gwiazdka}
{\Hta}(\mu;\Q) \geq {\hta}(\mu;\P)-\e. 
\end{equation}
Consequently, by Proposition \ref{propozycja} and the definition of Tsallis entropy, we have
\begin{equation}
{\hta}(\mu;\P) = {\hta}(\sum_{k=1}^n a_k \mu_k; \P) 
\end{equation}
\begin{equation}
\geq \sum_{i=1}^n a_i \hta(\mu_i;\P) \geq \sum_{i=1}^n a_i \Hta(\mu_i;\Q).
\end{equation}
Finally by (\ref{gwiazdka}), we obtain
\begin{equation}
{\Hta}(\mu;\Q) \geq {\hta}(\mu; \P) -\e 
\geq \sum_{i=1}^n a_i \Hta(\mu_i;\Q) - \e,
\end{equation}
which proves (\ref{nierowWaz}).

We prove the inequality (\ref{nierowWaz2}). For each $k \in \{1,\ldots n\}$ we find $\m^k \in \V(\mu_k;\Q)$ satisfying
\begin{equation} \label{2gwiazdki}
{\htva}(\mu_k; \m^k) \leq {\Hta}(\mu_k;\Q) + \frac{\e}{n}.
\end{equation}

Making use of Proposition \ref{propozycja} and (\ref{2gwiazdki}), we have
\begin{equation}
{\Hta}(\mu;\Q) \leq \sum_{i=1}^n a_i^\alpha \htva(\mu_i;\m)+ \frac{\sum_{k=1}^n a_k^\alpha - 1}{1-\alpha}
\end{equation}
\begin{equation}
\leq \sum_{i=1}^n a_i^\alpha \Hta(\mu_i;\Q)+ \frac{\sum_{k=1}^n a_k^\alpha - 1}{1-\alpha} + \e.
\end{equation}
This completes the proof as $\e>0$ was an arbitrary number.
\end{proof}

Let us observe that the estimation (\ref{nierowWaz}) and (\ref{nierowWaz2}) cannot be improved.

\begin{example} \label{exSharp}
We assume that $\alpha \in (0,\infty) \setminus \{1\}$, $X=\{0,1\}$ and $\mu_1, \mu_2$ denote discrete measures such that:
\begin{equation}
	\mu_1(\{0\})=1 \text{ and } \mu_2(\{1\})=1.
\end{equation}
Then, we have 
\begin{equation}
{\Hta}(a_1\mu_1+a_2\mu_2) = \frac{a_1^\alpha + a_2^\alpha - 1}{1-\alpha}.
\end{equation} 
It is exactly the right side of the inequality (\ref{nierowWaz2}). 

On the other hand, let $\mu_1, \mu_2$ be two measures which satisfy $\mu_1 = \mu_2$. Then 
\begin{equation}
{\Hta}(a_1\mu_1+a_2\mu_2) = {\Hta}(\mu_1)={\Hta}(\mu_2).
\end{equation}
It equals the right side of (\ref{nierowWaz}).
\end{example}

It is well-known that $\hta(\mu;\P) \to \hs(\mu;\P)$, when $\alpha \to 1$. Let us observe a similar relation between bounds obtained for Shannon entropy \cite[Theorem III.1]{Sm} and Tsallis entropy from Theorem \ref{corEnt}. Let us consider the functions:
\begin{equation}
l_\alpha(x,y) = a_1 x + a_2 y,
\end{equation}
\begin{equation}
u_\alpha(x,y) = a_1^\alpha x + a_2^\alpha y + \frac{ a_1^\alpha + a_1^\alpha - 1}{1-\alpha},
\end{equation}
which describe the lower and upper bound for the Tsallis entropy of order $\alpha$. If $x, y$ are non negative real numbers then these functions converge to the corresponding bounds calculated for Shannon entropy as $\alpha \to 1$, i.e.:
\begin{equation}
\left\{
   \begin{array}{ll}
		l_\alpha(x,y) \to a_1 x + a_2 y \\ 
		u_\alpha(x,y) \to a_1 x + a_2 y  - a_1 \log_2(a_1) - a_2 \log_2(a_2)\text{, }
	\end{array}
\right.
\end{equation}
when $\alpha \to 1$.

\section{Conclusion}

The weighted form of the entropy is very useful to derive properties of the entropy of the mixture of measures. We presented the condition under which every entropy function can be defined in the weighted way. The well-known Shannon, R\'enyi and Tsallis entropies satisfy this natural condition. We gave an example how use the weighted entropy to estimate the Tsallis entropy of order $\alpha$ of the mixture of measures. Obtained bounds are sharp and as a function of parameter $\alpha$, they converge to the corresponding bounds calculated for Shannon entropy. In similar manner, we can apply the tool of weighted entropy to compute for instance R\'enyi entropy of the mixture. 

\bibliographystyle{agsm}
\bibliography{entropy}

\end{document}